%% file: vmcai2013.tex
\documentclass{llncs}
\usepackage[ruled,lined,longend]{algorithm2e}
\usepackage[utf8]{inputenc}
\usepackage[usenames]{color}
\usepackage[table]{xcolor}
\usepackage{listings}
\usepackage{graphicx}
\usepackage{tikz}
\usepackage{amsmath, amssymb}
\usepackage{url}

\usetikzlibrary{arrows,automata}

\usepackage{macros}

\begin{document}

\title{%\vspace{-2cm}\large  
Black-Box Verification for GUI Applications}

\author{Stephan Arlt \and Evren Ermis \and Sergio
Feo-Arenis \and Andreas Podelski}
\institute{Albert-Ludwigs-Universit\"at Freiburg}

\maketitle

\input{src/abstract}

\input{src/introduction}

\input{src/motivation}

\input{src/approach}

\input{src/experiments}

\input{src/related}

\input{src/conclusion}
\clearpage
\bibliographystyle{abbrv}
\bibliography{vmcai2013}

\end{document}

%% file: src/abstract.tex
\begin{abstract}
  In black-box testing of GUI applications (a form of system testing),
  a dynamic analysis of the GUI application is used to infer a
  black-box model; the black-box model is then used to derive test
  cases for the test of the GUI application.  In this paper, we
  propose to supplement the test with the verification of the
  black-box model.  We present a method that can give a guarantee of
  the absence of faults, i.e., the correctness of \emph{all} test
  cases of the black-box model.  The  black-model
  allows us to formulate a parametrized verification problem. As we will show, it
  also allows us to circumvent the static analysis of the GUI tool
  kit.  We have implemented our approach;
  preliminary experiments indicate its practical potential.
% Recent efforts in GUI testing incorporate Event Flow Graphs (EFGs) to encode
% possible event sequences of GUI applications. In order to decide whether a GUI
% application satisfies a given property, event sequences of bounded length are
% selected from the EFG. Instead of selecting all bounded event sequences, we
% introduce the use of event flow guided static analysis. The key is to
% provide a suitable representation of the GUI application: We use an EFG to
% translate the GUI application into a model which encodes both the event flow and
% the event semantics. Static analysis is applied in order to output an event
% sequence which leads to the given property. The obtained event sequence is
% replayed on the GUI application to determine its feasibility. Infeasible
% sequences trigger an automatic refinement of the EFG until a feasible event
% sequence is found or the property is shown to be unreachable. We show the
% applicability of our approach by using benchmarks from the GUI testing
% community.
\end{abstract}

%%% Local Variables:
%%% TeX-master: "../vmcai2013"
%%% End:

%% file: src/introduction.tex
\section{Introduction and Overview}
\label{sec:introduction}

Many of today's applications provide a \emph{GUI} (Graphical User Interface)
which allows a user to interact with the system. A GUI consists of a set of
\emph{widgets}, e.g., buttons and check boxes. A user interacts with the system
by manipulating these widgets, e.g., clicking a button, or selecting a check
box. Each user interaction triggers an \emph{event}. An event is a message
generated by the GUI toolkit. A GUI application responds to a message with the
execution of an \emph{event handler} (e.g., a Java method).
%\andreas{letzter Satz noch korrekt?  ich habe ``which leads to'' mit ``with'' uebersetzt}

The GUI application must not be confused with the GUI toolkit. In the fashion of
an operating system, the GUI toolkit (through a \emph{message loop})
continuously listens for events and, each time an event is triggered, initiates
the execution of the corresponding event handler.

Until now, the topic of GUI applications as a target of formal, static
analysis-based verification has received little attention (see,
however,~\cite{chen02,dwyer04,ganov09} for initial attempts).  It is generally
agreed that a static analysis cannot adequately account for the GUI toolkit
(even its source code was fully available, which is generally not the case).
The GUI toolkit is, however, an integral part of a GUI application.  Without
accounting for the GUI toolkit, it already seems impossible to infer what
event sequences are executable (an event sequence models a user interaction).
Not every event sequence is executable; for example, an event may lead to
closing a window and thus disable another event.

In contrast, testing of GUI applications (which is a form of system testing) has
become a very active research topic; see
e.g.,~\cite{arlt12,gross12,mariani12,memon07,belli01,white00} and the references
therein. Within the setting of black-box system testing, a dynamic analysis of
the GUI application is used to infer a black-box model.
Subsequently,  the black-box model is used to derive the test cases.

In this paper, we base ourselves on the notion of the \emph{Event Flow
  Graph} (EFG), which is the particular black-box model in the popular
approach of~\cite{memon07}.  The nodes in an EFG are events.  Its edge translate
the is-followed-by relation between events.  The EFG represents a set of event
sequences (the labelings of the paths in the graph).

The EFG can be inferred automatically.  A so-called \emph{GUI Ripper}
systematically explores the space of possible executions of the GUI application
and records the is-followed-by relation between events during those executions
through the edges of the EFG.  Note that the exploration can never be guaranteed
to be exhaustive (this is inherent to the black-box setting; an event can be
missed, for example).  Thus, the EFG is not guaranteed to include every possible
event sequence of the GUI application.

In this paper, we take the view where the actual test of a GUI
application is separated from the construction of the EFG.  I.e., the
test refers to the previously inferred EFG.  The EFG serves to
prescribe a set of potential test cases; i.e., the EFG defines an
upper bound on the set of event sequences to be checked.  Continuing
this view, the EFG defines a specific verification problem: \emph{does
  there exist an event sequence in the given EFG that, if run on the
  GUI application, leads to a faulty execution?}\footnote{The problem
  formulation (as well as our use of the term \emph{completeness})
  reflects the perspective of testing, which is falsification rather
  than verification.  However, we view the problem as a verification problem.}

% To summarize, the test of a GUI application is parametrized by the previously
% inferred black-box model, here, the EFG. The EFG prescribes the potential test
% cases through its paths (of unbounded-length). In this paper, we propose to
% view the test as an (incomplete) solution to a specific verification problem.

% \begin{quote}
%   verification problem: \emph{does there exist an event sequence in
%     the given EFG that violates the specified correctness property of
%     the GUI application?}\footnote{To be precise, the formulation
%     indicates a falsification problem, which reflects the perspective
%     of testing.  We still call it a verification problem in order to
%     stress that here, we are interested in proving the absence of
%     certain faults.}
% \end{quote}

The test of the GUI application with test cases derived from the given
EFG provides an incomplete solution for the verification problem
(incomplete because at no point, no matter how many test cases are
derived from the EFG, the absence of faults can be guaranteed).  In
this paper, we present a complete solution to the verification
problem.

The motivation is to supplement testing.  Testing is known to be
indispensable but it comes with the notorious deficiency of being
incomplete.  At some point during the test-debug cycle, the tester may
desire a formal guarantee of the absence of faulty event sequences (in
the set prescribed by the given EFG).  Given such a guarantee, the
tester can stop testing: no test case derived from the given EFG will
yield to a positive.

As usual with an undecidable verification problem, the terminology
\emph{complete solution} refers to an abstraction-based verification
algorithm with or without iterated abstraction refinement, i.e., with
the possibility of non-termination or with the possibility of a don't
know answer.

From the perspective of the programmer, a GUI application consists of
the family of event handlers but, as explained above, verification
needs to also account for the effect of the GUI toolkit on the
executability of event sequences.  Thus, a static analysis-based
verification method seems a no-go.  However, and this is a basic
observation put forward in this paper, at least some of the
information about the effect of the GUI toolkit on the executability
of event sequences is encoded in the EFG (simply by the fact that the
EFG has been constructed by executing the GUI application).  

In this paper, we propose a way to make this information available.
Moreover, we can formally show that it is sufficient to solve our
verification problem. %, which refers to the given EFG.  Andreas26Aug
% \andreas{haben
  % wir in dem paper eine Formalisierung?  Stephan und ich hatten das
  % diskutiert: wir brauchen die Annahme, dass das GUI toolkit die
  % Variablen der event handlers nicht modifiziert. Unter dieser Annahme
  % gilt: das konstruierte Programm hat dieselbe Menge von event
  % sequences wie die GUI application eingeschraenkt auf die event
  % sequences des EFG.} 
  % \sergio{Im technischen Teil steht, dass wlog man annimmt, dass die Handlers
  % den Zustand des Programs ohne user-inputs zu beruecksichtigen aendern. Wlog
  % weil fuer solche Events man einfach ein Event aus jedem moeglichen Input-Wert
  % machen koennte. Da wuerde das Programm explosiv wachsen, sagen wir aber nicht.}

Given the EFG and the family of event handlers, one can construct a program that
 contains the event handlers and a \emph{mock-up} of the GUI toolkit (of its
message loop). The program mimicks the behavior of the GUI application as far it
is represented by the EFG. We can show that by applying a static analysis-based
verification method to the program, we obtain a complete solution for the
verification problem for the GUI application and the given EFG.

We will next discuss another issue, an issue that leads us to
enhancing the above verification method.  That is, not every event
sequence in the set represented by the EFG needs to be executable on
the GUI application.  This is because, although an event may be
executable after % Andreas26Aug
another one in one path%  (which gives rise to the
% corresponding edge in the EFG)
, this may no longer hold in another
path (for an example, see Section~\ref{sec:motivation}). Formalized in
the framework of \emph{abstract
  interpretation}~\cite{Cousot:1977:AIU:512950.512973}, the
construction of the EFG uses an abstraction of the (finite) set of
executed paths by a graph which denotes a larger (infinite) set of
paths.

The program that we have constructed as input for verification accounts for the
behavior of (the message loop of) the GUI toolkit according to the EFG.  I.e.,
not every execution of the program may be executable on the GUI application.  As
a consequence, if the verification of the program returns a counterexample (a
faulty execution of the program), then this counterexample may not be executable
on the GUI application.  In our setting of abstraction-based verification, we
thus have two notions of spuriousness.  The first spuriousness is due to the
abstraction used by the static analysis (the path may be \emph{infeasible}
because of data dependencies).  The second spuriousness is due to the use of a
graph (i.e., the EFG) as a representation of a set of paths (the path may be not
executable because of the behavior of the GUI toolkit).

There is no way to detect the non-executability statically (other than by the
static analysis of the GUI toolkit).  In the dynamic setting of testing,
however, the non-executability of an event sequence is detected \emph{by
nature}.  Having introduced a first heresy (by formally verifying a program that
is inferred from a black-box model), we may be as well introduce a second heresy
and use dynamic analysis in an abstraction refinement iteration.  I.e., if the
verification of the program detects a counterexample (i.e., an execution of the
program that violates the correctness condition) and if the dynamic analysis
determines the non-executability of the counterexample by the GUI application,
then, as we will shown, the program can be refined so that the counterexample is no longer a possible
execution of the new program.  We thus enhance the verification method
with a \emph{program refinement iteration}.   % Andreas26Aug

Even without the program refinement iteration, the verification of the program
constructed from a given EFG provides a guarantee of the absence of 
faults: no test case derivable from the EFG can lead to
a faulty execution of the GUI application.  % This is what
% we have called the completeness of the verification method (no faulty
% execution will be missed).
The enhanced verification method (with the program refinement
iteration) satisfies also the \emph{no false
  positives} requirement. I.e., if a
counterexample is returned then it indeed corresponds to an (``executable'')
faulty execution, and
if the GUI application does
not have a faulty execution, then the method will not return a
counterexample.
This means that we can use our method not only for verification but
also for falsification, i.e., as an enhancement of testing (it may
find errors that testing might not find).
% that we may call the \emph{soundness of falsification}
%(a property that testing satisfies \emph{by nature})
% Andreas26Aug

% We note that the addition of the program refinement iteration preserves the
% above-mentioned completeness property of verification (every faulty execution
% will be caught).  We also note that the program refinement iteration is
% guaranteed to terminate.  This is to be contrasted with the \emph{abstraction
% refinement} refinement in the CEGAR scheme.

% \andreas{We also note that the
% program refinement iteration is guaranteed to terminate. - Stimmt das?}
% \stephan{Nein, (ich vermute) das stimmt nicht.}

% To summarize, the algorithm with the program refinement iteration is a complete
% solution to a stronger version of the verification problem, namely: : \emph{does there exist an  \emph{executable} event sequence in the given EFG
%   that, if run on the GUI application, leads to a faulty
%   execution?}
% \begin{quote}
%   strong verification problem: \emph{given the EFG, does there exist
%     an \emph{executable} event sequence
%     in the EFG that violates the specified correctness property of the
%     GUI application?}
% \end{quote}

% \andreas{wenn wir Platz sparen muessen, dann koennen wir das ``quote''
% environment wegnehmen und die Formulierung inlinen}
% \andreas{sollen wir das paper auch so strukturieren: erst das eine und
%   dann das andere verification problem loesen?}

We have implemented our overall approach and we have evaluated it on a
number of existing benchmarks from GUI testing.  Our preliminary
experiments indicate a promising potential of our method as
supplement and enhancement to testing.

% moved to conclusion:
% The main contribution of the paper is a conceptual one: the realization that GUI
% testing that formal, static analysis-based verification methods can have an
% interesting potential for GUI applications.  The crux is to take the perspective
% of the GUI tester who, after having constructed the black-box model of the GUI
% application, is faced with what can be phrased as a verification problem.

% The technical contribution of the paper is to \emph{make things work}.
% This includes the construction of the program from a given EFG as input to a
% static analysis-based verification procedure and the program refinement
% procedure.

%%% Local Variables:
%%% TeX-master: "../vmcai2013"
%%% End:

%% file: src/motivation.tex
\section{Motivating Example}
\label{sec:motivation}

In this section we illustrate the application of our approach on an example of a
GUI application. The example application depicted in the screenshot in
Figure~\ref{fig:example} consists of a \texttt{MainWindow} and a
\texttt{Dialog}. The \texttt{MainWindow} contains three buttons that can fire
the events $e_1$, $e_2$, and $e_3$. The \texttt{Dialog} contains one button
which can fire event $e_4$.

\begin{figure}
	\centering
	\includegraphics[scale=.4]{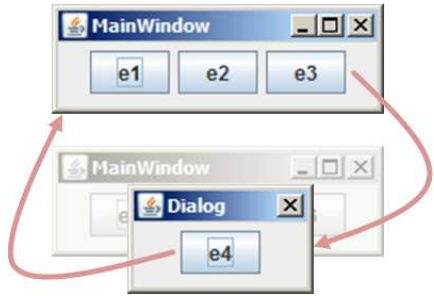}
	\caption{An example of a GUI application. The arrows between the two screenshots
	indicate the transition between two views: Clicking the button for event $e_3$
	leads to the \texttt{Dialog}; i.e., the buttons for the events $e_1$, $e_2$, and
	$e_3$ are no longer enabled. Clicking the button for event $e_4$ closes the
	\texttt{Dialog} and leads back to the first view; i.e., the buttons for the
	events $e_1$, $e_2$, and $e_3$ are enabled again.}
	\label{fig:example}
\end{figure}

In the first step of our approach we incorporate an Event Flow Graph (EFG) in
order to model possible user interactions of the GUI application. An EFG is a
directed graph where each node represents an event of the GUI. An edge between
two events states that the corresponding events can be executed consecutively.
The EFG of the example application depicted in Figure~\ref{fig:example-efg}
consists of the four events $e1$, $e2$, $e3$, and $e4$. Such a graph represents
the black-box model that is used to generate test cases in~\cite{memon07}. The
idea is that a path in the EFG encodes a sequence of user interactions. The
marking of $e_1$, $e_2$, and $e_3$ as initial nodes encodes how a user
interaction can start. In our approach we automatically infer an EFG using
reverse engineering~\cite{memon03}. Moreover, we apply our tool, called
\emph{Gazoo}~\cite{arlt12}, in order to extract the event handlers of each
event.

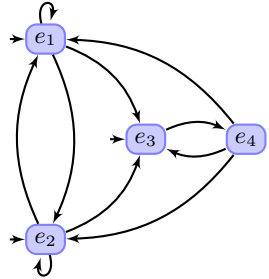
\begin{figure}
	\centering
	\input{pics/efg}
	\caption{Event Flow Graph for the example application. Each event can follow
	each other event except that $e_4$ can follow only after $e_3$. A path in the EFG encodes a
	sequence of user interactions. The marking of initial nodes encodes that a user
	interaction can start with $e_1$, $e_2$, $e_3$, but not $e_4$.}
	\label{fig:example-efg}
\end{figure}

The set of event handlers and the EFG serve as input for the translation of the
GUI application into a verifiable program. That is, we apply our tool, called
\emph{Joogie}~\cite{joogie}, in order to construct the program depicted in
Figure~\ref{lst:example-boogie} which represents the event handlers and a
\emph{mock-up} of the message loop as far it is represented by the EFG. The
program starts with the block \texttt{START} which initializes the used
variables. Furthermore, the block \texttt{START} provides a \texttt{goto}
statement which allows the non-deterministic choice of the blocks \texttt{e1},
\texttt{e2}, and \texttt{e3}. The blocks that can be chosen within the block \texttt{START}
conform to the initial events that can be chosen from the EFG. In this program,
each event handler is encoded as a set of blocks. For example, the event $e_1$
is encoded in block \texttt{e1}, and event $e_2$ is encoded in the blocks
\texttt{e2}, \texttt{e2\_THEN}, \texttt{e2\_ELSE}, and \texttt{e2\_ENDIF}. The
last block of each event handler contains a \texttt{goto} statement which allows
a non-deterministic choice of possible succeeding events according to the EFG.
The block \texttt{EXIT} contains a property, encoded as an assertion, which must
hold after the execution of each event handler (the \texttt{goto} statement in
each event handlers provides the possibility of choosing the block
\texttt{EXIT}).

\begin{figure}
	\centering
	\input{lst/example-boogie}
	\caption{A code snippet of the constructed (Boogie) program (of the EFG and the
	extracted event handlers). The procedure \texttt{EFG\_Procedure} defines a set
	of blocks for the event handlers $e_1$, $e_2$, $e_3$, and $e_4$. The event
	handler $e_1$ assigns the constant value $1$ to variable \texttt{x} (line~8).
	The event handler $e_2$ multiplies the value of variable \texttt{x} with $2$
	(line~12). If the value of variable \texttt{x} is greater than $4$, then the
	event handler $e_2$ disables the button which can fire the event $e_3$
	(line~17). The event handler $e_3$ decrements variable \texttt{x} and
	opens the dialog (line~28-29). The event handler $e_4$ assigns the constant
	value $1$ to variable \texttt{x} and closes the dialog (line~33-34).}
	\label{lst:example-boogie}
\end{figure}

In the third step of our approach we apply a static analysis, called
\mc~\cite{ermis12}, on the program constructed in the previous step. We inject
different assertions (that is, properties of the GUI application) and show the
outcome of the static analysis.

If we inject the assertion \texttt{assert(x != 7)}, the static analysis outputs
the result \texttt{UNSAFE}, that is, a violation of this assertion is found in
the program. In particular, the static analysis outputs the shortest sequence of
events which leads to the violation, namely the event sequence
\es{$e_1$}~\es{$e_2$}~\es{$e_2$}~\es{$e_2$}~\es{$e_3$}.

In order to validate whether this event sequence is actually executable, we
replay the event sequence on the GUI application using a replayer. That is, we
apply the tool \emph{GUITAR}~\cite{memon07}, which takes as input a sequence of
events to mimmick user interactions in a black-fashion on the GUI application.
The result of the replayer states, whether the event sequence was successfully
replayed or not.

In the case of the assertion \texttt{assert(x != 7)}, the event sequence
\es{$e_1$}~\es{$e_2$}~\es{$e_2$}~\es{$e_2$}~\es{$e_3$} is not executable on the
GUI. The reason is that the last event \es{$e_2$} in the event sequence disables
the event \es{$e_3$} (which refers to line~17) in
Figure~\ref{lst:example-boogie}. The fact that the event sequence is not
executable causes our approach to automatically refine the used EFG.

The refinement step of our approach modifies the EFG, such that the detected
non-executable event sequence (encoded as a path in the graph) is not contained
in a refined EFG. In order to achieve this goal, we first convert both the EFG
and the non-executable event sequence into automata, which allows us to apply
operations from regular languages on these automata. In particular, the
non-executable event sequence is encoded as an accepting word of the EFG
automaton. We construct the complement of the accepting word and intersect it
with the EFG automaton. The result is a refined EFG automaton which does not
accept the non-executable event sequence. Finally, we convert the refined EFG
automaton back into an EFG, which we call \emph{Extended EFG} as depicted in
Figure~\ref{fig:eefg}.

The extended EFG is now used to construct a new program. The static analysis is
again applied to verify whether the assertion \texttt{assert(x != 7)} is
violated. In this case, the output of the static analysis is \texttt{SAFE},
since there exist no sequence of events in the program which leads to the
violation of the property.

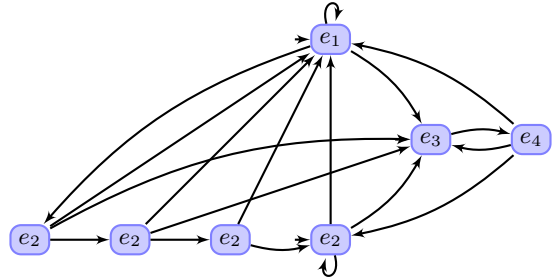
\begin{figure}
	\centering
	\input{pics/eefg}
 	\caption{Extended Event Flow Graph for the example GUI. Comparing to the
 	original EFG, the event sequence $e_1, e_2, e_2, e_2, e_3$ is no longer
 	possible.}
	\label{fig:eefg}
\end{figure}

For completeness we inject one further assertion: For the property
\texttt{assert(x != 3)}, the static analysis outputs the result \texttt{UNSAFE}
and the event sequence \es{$e_1$}~\es{$e_2$}~\es{$e_2$}~\es{$e_3$} which leads
to the violating state. Unlike our first example, this event sequence is indeed
executable on the GUI application without a further iteration of the refinement
loop.

%% file: pics/efg.tex
\begin{tikzpicture}[auto,node distance=1.333cm,->,>=latex',thick,bend
angle=25,/tikz/initial text=,scale=.5]
	\tikzstyle{event}=[rectangle,draw=blue!50,fill=blue!20,rounded corners]

	\node[event,initial] (e1) [] {$e_1$};
	\node[event,initial] (e3) [right of=e1,below of=e1] {$e_3$};
	\node[event,initial] (e2) [left of=e3, below of=e3] {$e_2$};
	\node[event] (e4) [right of=e3] {$e_4$};
	
	\path (e1) edge[loop above] node[] {} (e1)
	      (e1) edge[bend left] node[] {} (e2)
	      (e1) edge[bend left] node[] {} (e3)
	      (e2) edge[loop below] node[] {} (e2)
	      (e2) edge[bend left] node[] {} (e1)
	      (e2) edge[bend right] node[] {} (e3)
	      (e3) edge[bend left] node[] {} (e4)
	      (e4) edge[bend right] node[] {} (e1)
	      (e4) edge[bend left] node[] {} (e2)
	      (e4) edge[bend left] node[] {} (e3);
\end{tikzpicture}

%% file: lst/example-boogie.tex
\begin{minipage}{.9\textwidth}
\begin{lstlisting}[language=Boogie]
procedure EFG_Procedure() modifies ...
{
	INIT: 
	...
	START:
		x := 0;
		goto e1, e2, e3;

	e1: // handler of event e1
		x := 1;
		goto e1, e2, e3, EXIT;
		
	e2: // handler of event e2
		x := x * 2;
		goto e2_THEN, e2_ELSE;
		
	e2_THEN:
		assume (x > 4);
		call e3$setEnabled(false); // disable event e3 
		goto e2_ENDIF;
		
	e2_ELSE:
		assume (!(x > 4)); 
		goto e2_ENDIF;
		
	e2_ENDIF:
		goto e1, e2, e3, EXIT;

	e3: // handler of event e3
		x := x - 1;
		call dialog$setVisible(true);
		goto e4, EXIT;

	e4: // handler of event e4
		x := 1;
		call dialog$setVisible(false);
		goto e1, e2, e3, EXIT;
		
	EXIT:
		assert (x != 7); // assertion to be checked
		return;
}
\end{lstlisting}
\end{minipage}

%% file: pics/eefg.tex
\begin{tikzpicture}[auto,node distance=1.333cm,->,>=latex',thick,bend
angle=15,/tikz/initial text=,scale=.5]
	\tikzstyle{event}=[rectangle,draw=blue!50,fill=blue!20,rounded corners]

	\node[event,initial] (e1) [] {$e_1$};
	\node[event,initial] (e3) [right of=e1,below of=e1] {$e_3$};
	\node[event,initial] (e2) [left of=e3, below of=e3] {$e_2$};	
	\node[event] (e23) [left of=e2] {$e_2$};
	\node[event] (e22) [left of=e23] {$e_2$};
	\node[event] (e21) [left of=e22] {$e_2$};
	\node[event] (e4) [right of=e3] {$e_4$};
	
	\path (e1) edge[loop above] node[] {} (e1)
	      (e1) edge[bend right] node[] {} (e21)
	      (e1) edge[bend left] node[] {} (e3)
	      (e21) edge[] node[] {} (e1)
	      (e21) edge[] node[] {} (e22)
	      (e21) edge[bend left] node[] {} (e3)
	      (e22) edge[] node[] {} (e1)
	      (e22) edge[] node[] {} (e23)
	      (e22) edge[] node[] {} (e3)
	      (e23) edge[] node[] {} (e1)
	      (e23) edge[bend right] node[] {} (e2)
	      (e2) edge[loop below] node[] {} (e2)
	      (e2) edge[] node[] {} (e1)
	      (e2) edge[bend right] node[] {} (e3)
	      (e3) edge[bend left] node[] {} (e4)
	      (e4) edge[bend right] node[] {} (e1)
	      (e4) edge[bend left] node[] {} (e2)
	      (e4) edge[bend left] node[] {} (e3);

\end{tikzpicture}

%% file: src/approach.tex
\section{Approach}
\label{sec:approach}

The runtime behavior of event-driven applications is determined by the
occurrence of events. In our context, events can be triggered by the user
interacting with the graphical user interface of the application. The GUI
restricts the events that may be executed at any given time by, for example,
enabling and disabling controls or displaying modal dialog boxes. An event flow
graph effectively encodes those restrictions.

The goal of a tester is to determine whether, among the set of event sequences
encoded by the event flow graph, all those that are \emph{executable} in the
application satisfy a property of interest. That is, whether they are
\emph{correct}. It is evident that testing cannot efficiently provide a
guarantee that all test sequences are correct. We thus propose to encode the set
of sequences represented by the EFG as a program, which traces correspond to
event sequences of the EFG. The program can then be statically analyzed to
efficiently determine whether all traces are correct, from which we can deduce
that all sequences of the EFG are also correct. If a violating sequence is
found, we use a replayer to dynamically check whether the sequence found is an
executable one. In order to automate the process, a non-executable event
sequence triggers a refinement loop where the EFG is modified to exclude all
event sequences that start with the one found by the static analyzer.

The verification algorithm (see Algorithm~\ref{alg:verify}) of our approach
takes as input a system $S$ and an event flow graph $E$. The system $S$ contains
the GUI application to be tested (including the properties to be checked), the
underlying GUI toolkit, the operating system, etc. The EFG $E$ can be created
manually or automatically~\cite{memon03}. The algorithm starts with the
extraction of the set of event handlers using the system $S$ and the EFG $E$;
see line~2 and Section~\ref{sec:ripper}. Then, the algorithm constructs a
program $P_{EFG}$ which represents a mock-up of the message loop by employing
the set of event handlers $H$ and the EFG $E$; see line~4 and
Section~\ref{sec:translation}. In line~5, a static analysis is applied which
returns a sequence of events $\seq$; see Section~\ref{sec:modelchecker}. If the
event sequence $\seq$ is empty, then the static analysis did not find a
violation of a property, and \textit{success} is returned from the verification
algorithm (line~14). In contrast, if the event sequence $\seq$ is not empty,
then the static analysis found an event sequence which violates a property.
Thus, the violating event sequence $\seq$ is replayed on the system $S$; see
line~7, Section~\ref{sec:replayer}. The replayer returns a boolean flag
indicating the executability of the sequence, and the non-executable prefix
$\seq_\infeasible$ of the sequence $\seq$. If $\seq$ was executable, the
verification algorithm returns \textit{fail} (line~11). If $\seq$
was not executable, the non-executable prefix $\seq_\infeasible$ of the event
sequence $\seq$ is removed from the EFG; see line~9 and
Section~\ref{sec:refinement}. The result of the refinement is an extended EFG,
which is used to construct a new program $P_{EFG}$.

\input{algs/verify}

We start by formalizing the problem of verification of GUI applications. We then
describe the details of the steps involved in our proposed algorithm.

\input{src/prelim}

\input{src/ripper}
\input{src/translation}

\input{src/modelchecker}

\input{src/replayer}

\input{src/refinement}

%% file: algs/verify.tex
\LinesNumbered
\begin{algorithm}
\KwIn{$S$ : System, $E$ : EFG}
\KwOut{(\texttt{fail}, $\seq$) or (\texttt{success}, \textit{null}) 
  \mbox{ \ \ \ \ \  ($\seq$ is \textit{null}} or an event sequence in $E$)}
% (including failing event sequence $\seq$)}
\DontPrintSemicolon
\Begin{
	$H$ := ExtractEventHandlers($S$, $E$) \;
	\While{true}{
		$P_{EFG}$ := BuildMessageLoop($H$, $E$) \;
		$\seq$ := RunStaticAnalysis($P_{EFG}$) \;
		\eIf{$\seq \neq \textit{null}$}{
			$(IsExec, \seq_\infeasible)$ := ReplaySequence($S$, $\seq$) \;
			\eIf{$IsExec$ = false}{
				$E$ := RefineEFG($E$, $\seq_\infeasible$) \;
			}{
				\Return (\texttt{fail}, $\seq$) \;
			}
		}{
			\Return (\texttt{success}, \textit{null}) \;
		}
	}
}
\caption{\emph{Verification} Algorithm.}
\label{alg:verify}
\end{algorithm}

%% file: src/prelim.tex
\subsubsection{Extended Flow Graphs}
\label{sec:prelim}
In the literature~\cite{memon07}, an event flow graph is defined
as a directed graph $G = \langle \eventset, \initset, \efgedge \rangle$ where
$\eventset$ is the set of events, $\initset \subseteq \eventset$ is the set of
initial events and $\efgedge \subseteq \eventset \times \eventset$ is the event
flow relation.
An edge $(\event,\event') \in \efgedge$ between two events $\event, \event' \in
\eventset$ states that the event $\event'$ can be executed after the event
$\event$.
If there is no edge between events $\event, \event'$ then event $\event'$ cannot
be executed after event $\event$.

In order to support our refinement approach, we extend the classical definition
of event flow graphs to enable having multiple nodes labeled with the same
event.
\begin{definition}[Extended Event Flow Graph]
\label{def:eefg}
An extended event flow graph is a labeled, directed graph $\EFG = \langle
\locations, \initloc, \efgedge, \labeling \rangle$ where $\locations$ is the set of locations, $\initloc$ the
set of initial locations, $\efgedge \subseteq \locations \times \locations$ the
transition relation and $\labeling: \locations \mapsto \eventset$ a labeling
function that assigns an event from the set of events $\eventset$ to every
location in $\locations$. 
\end{definition}
We additionally define the set $\paths(\EFG)$ as the set of sequences of
the form $\pi=\location_0,\location_1,\ldots,\location_n$ where $\location_0 \in
\initloc$, $\location_1,\ldots,\location_n \in \locations$ and
$\forall i \in \{0,\ldots,n-1\}: (\location_i,location_{i+1}) \in \efgedge$.

An extended event flow graph (EEFG) thus encodes a set of event sequences that
are considered \emph{possible} in an application. That is, that the user
interactions represented by the event sequence can be observed during the
execution of the application.

\begin{definition}[Possible Event Sequence]
\label{def:eventseq}
Let $\mathit{Paths}(\EFG)$ be the set of all paths in the extended event flow
graph $\EFG = \langle \locations, \initloc, \efgedge, \labeling \rangle$. A
sequence of events $\sigma = e_0,e_1,\ldots,e_n$ is called a \textbf{possible
event sequence} if and only if there exists a path
$\pi=\location_0,\location_1,\ldots,\location_n \in \paths(\EFG)$ such that
$\forall i \in \{0,\ldots,n\}: \labeling(\location_i) = e_i$.
\end{definition}
For brevity, we denote all possible event sequences for a given EEFG $\EFG$ with
the set $\possible(\EFG)$.

\subsubsection{Verification of GUI Applications}
As usual, we define the states of an application as pairs $\progstate = \langle
\location,\valuation \rangle$ where $\location$ is the value of the program
counter (the program \emph{location}) and $\valuation$ is a valuation function
that assigns values to the program's variables in their corresponding domains. There is a special state
denoted by $\initstate = \langle \hat{\location_0}, \hat{\valuation_0}\rangle$
that is the state corresponding to the application's entry point and the
initial variable valuation.

An application trace $\trace = \progstate_0,\progstate_1,\ldots$ is a sequence
of states where $\progstate_0 = \initstate$ and every $\progstate_i = \langle
\location_i,\valuation_i \rangle$, $i>1$ is the effect of executing the
instruction at location $\location_{i-1}$ under valuation $\valuation_{i-1}$. We
assume the usual operational semantics of program instructions.

Some intermediate states $\progstate_i$ in a trace correspond to the entry
points of \emph{event handlers}. Event handlers are program functions that are
executed whenever an interaction with the graphical interface (an \emph{event})
is triggered. We define the \emph{handler map} function $\handlermap: \eventset
\mapsto \identifier$ where $\eventset$ is the set of events of the graphical
interface and $\identifier$ the set of locations that correspond to the entry
points of the functions of the application.

For simplicity, and without loss of generality, we consider only events whose
handlers change the program state without considering user inputs such as, e.g.,
scroll bars or text boxes. Those events that read user inputs can be replaced by
a family of events where there is one for every possible input value.
With that consideration we can construct the trace that corresponds to an event
sequence.
\begin{definition}[Traces from event sequences] \label{def:tracefromseq}
Let $\seq = e_0,e_1,\ldots,e_n$ be an event sequence, the corresponding trace is
\[\trace_\seq = \initstate,
\progstate_{init_1},\ldots,\progstate_{e_0}=\langle\handlermap(e_0),
\valuation\rangle,\progstate_{0,1},\ldots,
\progstate_{e_1}=\langle\handlermap(e_1),
\valuation\rangle,\progstate_{1,1},\ldots \] where the states
$\progstate_{init_1}$ are part of the \emph{initialization sequence} of the
application, states $\progstate_{e_i}$ are the states at the entry point of the
handler of event $i$ and the states $\progstate_{i,j}$ are the states resulting
from executing instruction $j-1$ of the event handler $\handlermap(e_i)$.
\end{definition}

In our setting, \emph{assertions} are special instructions of the form
\texttt{assert}$(\mathit{expr})$ that have no effect on valuations and simply
update the program counter to the following instruction. Arguments to assertions are
boolean formulas over the application variables. A state $\progstate
= \langle \location,\valuation \rangle$ satisfies the assertion $a$ at location
$\location_a$ if and only if $\location = \location_a$ and the interpretation of
$\mathit{expr}$ under valuation $\valuation$ is $\mathit{true}$ (denoted
$\progstate \models a$). We also introduce $\assertions(\app)$ to
denote the set of assertions in the code of application $\app$.

An event sequence $\seq$ is called \emph{correct} if and only if every state
$\progstate_i$ in $\trace_\seq$ satisfies all assertions in the application's
code.

A GUI application, due to the restrictions imposed by its graphical interface,
defines a set of event sequences that can be executed. We denote the set of
executable event sequences of an application $\app$ with $\exec(\app)$. We thus
redefine the notion of correctness. 

\begin{definition}[Application Correctness] \label{def:correctness}
A GUI application $\app$ is considered \textbf{correct} with
respect to an EEFG $\EFG$ if and only if all possible event sequences of the
event flow graph that are executable, are correct. I.e.,
\[
\forall \seq \in \possible(\EFG) \cap \exec(\app) : \forall a \in
\assertions(\app) : \seq \models a
\]
\end{definition}
We thus aim to provide a guarantee of correctness with respect to an EEFG by
applying the algorithm steps described in detail in the following sections.

%% file: src/ripper.tex
\subsection{Extracting Event Handlers}
\label{sec:ripper}

Each event in a GUI (e.g., a click on a button \texttt{OK}) is encoded as an
\emph{event handler} (e.g., \emph{OnClickOK}). The step Event Handler Extraction
mixes aspects of black-box and white-box approaches in order to extract the
GUI's event handlers. First, as in a black-box approach, we execute the GUI
application and enumerate the GUI's widgets (e.g., windows, buttons, and text
fields). Then, leaving a pure black-box approach, we apply
Reflection\footnote{\url{http://java.sun.com/developer/technicalArticles/ALT/Reflection/}}
in order to obtain the Java object corresponding to each widget (e.g.,
\texttt{JButton}. We ask the Java object to invoke the method
\texttt{getActionListeners}. The method invocation returns the widget's event
handlers (which are called \emph{action listeners} in Java). In order to have a
unique identifier for the event, we use the \emph{widget ID} of the GUI
application. If the Java object does not provide the method
\texttt{getActionListeners}, there exists no registered event handler to this
widget. In this case, we simply discard the handler event for this event.
However, in the program construction the is-followed-by relation between events
is considered.

\subsubsection{Discussion}
We use a dynamic approach to extract the event handlers. In principle it would
be possible to extract these event handlers via a static approach by analyzing
the source code. However, since GUI code is written in so many ways, looking for
\emph{ActionListeners} did not suffice as the event handlers might also be
registered using callbacks, virtual function calls, or even external resource
files.

%% file: src/translation.tex
\subsection{Program Transformation}
\label{sec:translation}
% If complete non-determinism is assumed when verifying event-driven
% applications, many spurious violations of a safety property may be found, namely
% those that arise when checking the occurrence of event sequences that are not
% possible during the applications' execution. It is thus reasonable to aim at
% reducing the number of false-positives by incorporating the EFG of an
% application in the model checking process. 
In order to encode the testing problem as a verification task, we construct a
Boogie~\cite{boogie} program that incorporates the information of an EFG to
simulate the message loop of the underlying GUI toolkit.  

Given an EEFG $\EFG = \langle \locations, \initloc, \efgedge, \labeling
\rangle$, the program $\pefg$ is constructed by first inlining all
initialization functions to make sure the application's variables have their
initial values set correctly before simulating the message loop of the GUI
toolkit. Subsequently, entry and exit labels for the message loop are added to
the program. For each location $\location$ in the EFG, the corresponding event
handler $\handlermap(\labeling(\location))$ is inlined. At the end of the block,
a non deterministic \emph{goto} statement is added. The target labels of the
\emph{goto} statement are the labels of the locations in the set $\{\location_i
\suchthat (\location,\location_i) \in \efgedge \} \cup \{\mathit{EFG\_Exit}\}$.
The jump to the exit label is necessary to simulate event sequences of finite
length in the case of a cyclic EFG. Assertions in the application code are
translated as assertions in $\pefg$.

To better illustrate the program transformation, the control flow graph of the
program generated for the example in Section~\ref{sec:motivation} is shown in
Figure~\ref{fig:pefg}, a simplified version of the source code is shown in
Figure~\ref{lst:example-boogie}. Note the similarity of the control flow graph
of the generated program and the original EFG (see
Figure~\ref{fig:example-efg}).
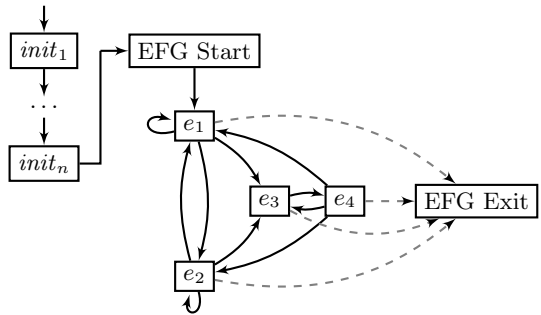
\begin{figure}
\centering
\input{pics/pefg}
\caption{Control flow graph of the program generated for the example of
Section~\ref{sec:motivation}. The nodes correspond to the inlined control flow
graphs of the initialization functions and the event handlers.}
\label{fig:pefg}
\end{figure}

\begin{theorem}
Let $\EFG$ be an EEFG and $\traces(\pefg)$ be the set of traces feasible in the
program $\pefg$. For every possible event sequence in $\EFG$, the corresponding application trace is feasible in the program $\pefg$. I.e.,
\[
\forall \seq \in \possible(\EFG): \trace_\seq \in \traces(\pefg)
\]
\end{theorem}

\begin{proof}
(sketch) For an arbitrary sequence $ \seq = e_0,e_1,\ldots,e_n \in
\possible(\EFG)$, use Definition~\ref{def:tracefromseq} to produce the corresponding trace
$\trace_\seq = \progstate_0, \progstate_1, \ldots, \progstate_m$. By the
construction of $\pefg$, define a mapping from states of $\pefg$ to
states in the application. Also map instruction labels in the application
correspond to instruction labels in $\pefg$. 

Assuming $\progstate_0 = \initstate$, every subsequent application state is
justified by an instruction in the application, show by the construction of
$\pefg$ that the mapped program state is justified by the corresponding
instruction in $\pefg$. Split into cases for the initialization phase and the
execution of event handlers.
\end{proof}

\begin{corollary}\label{cor:violatingtrace}
If a trace $t$ in $\traces(\pefg)$ violates an assertion $a$ (denoted $t
\models a$, the corresponding event sequence calculated using the mappings from
the previous proof and the inverse of Definition~\ref{def:tracefromseq} is not correct.
\end{corollary}

In order to solve our problem of interest (see
Definition~\ref{def:correctness}), we apply static analysis to determine whether
a trace that violates an assertion exists. By
Corollary~\ref{cor:violatingtrace}, the result of the analysis solves the
problem
\[
\forall \seq \in \possible(\EFG) : \forall a \in \assertions(\app) : \seq
\models a
\]
If a violating event sequence is found, it only remains to check whether it is
executable.

\subsubsection{Implementation}
We implemented our translation approach for GUI applications written in Java, we
integrate the generation of models as an option in Joogie~\cite{joogie}, which
translates Java bytecode into Boogie. Our implemented translation is available for download
at \url{joogie.org}.

%% file: pics/pefg.tex
\begin{tikzpicture}[auto,->,>=latex',thick,bend
angle=15,/tikz/initial text=,scale=.5]
	\tikzstyle{event}=[rectangle,draw=black]
	\tikzstyle{every node}=[node distance = 0.75cm]
	\node[coordinate] (start) {};
	\node[event] (init1) [below of=start, node distance=0.6cm] {$\mathit{init}_1$};
	\node[] (init2) [below of=init1] {$\ldots$};
	\node[event] (inits) [below of=init2] {$\mathit{init}_n$};
	\node[event] (efginit) [right of=init1, node distance = 2cm] {EFG Start};
	\tikzstyle{every node}=[node distance = 1cm]
	\node[event] (e1) [below of = efginit] {$e_1$};
	\node[event] (e3) [right of=e1,below of=e1] {$e_3$};
	\node[event] (e2) [left of=e3, below of=e3] {$e_2$};
	\node[event] (e4) [right of=e3] {$e_4$};
	\node[coordinate] (efginitanchor) [left of=efginit, node distance = 1.25cm] {};
	\node[event] (efgexit) [right of=e4, node distance = 1.75cm] {EFG Exit};
	
	\path 	(start) edge[] node[] {} (init1)
			(init1) edge[] node[] {} (init2)
			(init2) edge[] node[] {} (inits)
			(efginit) edge[] node[] {} (e1)
	      ;
	\path[draw,-] (inits) -| (efginitanchor) edge[->] node[] {} (efginit);
	\path (e1) edge[loop left] node[] {} (e1)
	      (e1) edge[bend left] node[] {} (e2)
	      (e1) edge[bend left] node[] {} (e3)
	      (e2) edge[loop below] node[] {} (e2)
	      (e2) edge[bend left] node[] {} (e1)
	      (e2) edge[bend right] node[] {} (e3)
	      (e3) edge[bend left] node[] {} (e4)
	      (e4) edge[bend right] node[] {} (e1)
	      (e4) edge[bend left] node[] {} (e2)
	      (e4) edge[bend left] node[] {} (e3);
	\tikzstyle{every edge}=[dashed,->,draw=black!50, bend angle=25, >=latex']
	\path 	(e1) edge[bend left] node[] {} (efgexit)
			(e2) edge[bend right] node[] {} (efgexit)
			(e3) edge[bend right] node[] {} (efgexit)
			(e4) edge[] node[] {} (efgexit)
			;
\end{tikzpicture}

%% file: src/modelchecker.tex
\subsection{Static Analysis}
\label{sec:modelchecker}
In the following we will explain how static analysis processes the produced
Boogie program to verify its safety wrt. assertions or to retrieve an event
sequence $\seq_\infeasible$.
For our purposes we need a tool that can handle complex looping structures in
programs well. In the case that the Boogie program has a feasible error trace,
the tool should return the shortest trace that leads to the error, thus ensuring
that the biggest possible trace sets are excluded from the EFG in every
refinement step.

Additionally, the counterexample must be complete in the sense that it must
document the error trace beginning from the initial location of the Boogie
program to the failing assertion. We use the model checker \mc~(a loose
homophone of ``codecheck'') which is based on \textsc{Ultimate}~\cite{ermis12}.
\mc~uses a splitting algorithm with large block encoding(LBE)~\cite{beyer09}.
This combination has proven itself to be very efficient on programs with complex
looping structures in program code. \mc~uses a breadth first search on the
program. Therefore it always finds the shortest failing trace in the program.
The interpolating splitting algorithm does not guess overapproximations of
loops. It rather constructs an invariant for each loop. The invariant is
constructed by splitting the loop using interpolation. A possible counterexample
is encoded in a FOL formula and passed over to the interpolating theorem prover
SMTInterpol~\cite{DBLP:conf/spin/ChristHN12}. If the formula is unsatisfiable
the model checker obtains Craig interpolants for each node on the spurious
counterexample and splits each node by partitioning its state space. On loops
the splitting can be considered as an unwinding. Hence, due to the splitting,
the counterexample is always complete in our sense and the corresponding trace
is easy to comprehend from initial location to error location. As such
\mc~fullfills all our requirements.
\paragraph{From Error Trace to Event Sequence.}
The model checker \mc~constructs the control flow graph(CFG) of the
Boogie program. As a preprocessing step the entire program is transformed into a
simple goto-program. This eases the mapping of control locations in the code to
the CFG. Each label becomes a node in the CFG.
All other statements are encoded as transitions between those nodes.
There exist 3 possible outcomes, (1) \textsc{Safe}, (2) \textsc{Unsafe}, and (3)
\textsc{Unknown} (Figure~\ref{fig:mcresults}).
\begin{figure}[h!] \centering
\begin{tabular}{l|p{10cm}l}
(1) \textsc{Safe} & \mc~has proved that the specified assertion is not violated
in the program. Hence, there is no error trace to derive an event sequence from.\\
\hline
(2) \textsc{Unsafe} & \mc~has found a counterexample and its corresponding
error trace.
\mc~will derive the event sequence $\seq_\infeasible$ from the error trace and
report that the assertion has been violated. The event sequence is returned as
the witness of the error.\\
\hline
(3) \textsc{Unknown} & \mc~is neither able to verify the program nor to find a
feasible counterexample. This can be caused e.g. by an unknown result from the
theorem prover or a timeout etc.
\end{tabular}
\caption{Result types of the model checker \mc~and the according returned
information to the Replayer.\label{fig:mcresults}}
\end{figure}
The result types (1) and (3) stop the computation and the overall algorithm. If
the result is \textsc{Unsafe}, the event sequence $\seq_\infeasible$ must be
derived from the error trace. The error trace is an alternating sequence of
nodes and transitions. According to our prior explanation, each node corresponds
to a label in the Boogie program. Specially marked labels of the Boogie program
correspond to event handlers. Hence, we can easily derive the event sequence
$\seq$ by mapping each marked label on the error trace back to its original
event. This event sequence is handed over to the Replayer.

%% file: src/replayer.tex
\subsection{Replaying an Event Sequence}
\label{sec:replayer}

The replayer takes as input a sequence of events and embeds it into a GUI test
case. A GUI test case consists of four components: (1) a \emph{precondition}
that must hold before executing an event sequence; (2) the \emph{event sequence}
to be executed; (3) possible \emph{input-data} to the GUI; and (4) a
\emph{postcondition} that must hold after executing the event sequence. The step
Replayer ensures the precondition, executes the event sequence on the GUI
application, inserts input data where necessary, and checks if the postcondition
holds. For (1), as a precondition we define that all user settings of an AUT
have to be deleted before executing the event sequence.
For (3), we generate random data, i.e., random strings for text boxes. The
computation of suitable input values (see ~\cite{ganov09}) for widgets
represents an orthogonal problem and is not in the scope of this paper. For (4),
we use an \emph{executability monitor} as an oracle, that is, the GUI test case
is marked as passed if the event sequence is executable on the GUI application.
The GUI test case is marked as failed if the event sequence is not executable on
the GUI application. Furthermore, the replayer recognizes which event in the
event sequence is not executable anymore. Then, the replayer stops the execution
and extracts the shortest non-executable prefix of the event sequence. The
replayer forwards the prefix (a sequence of events) to the refinement procedure,
which removes the prefix from the used EFG.

%% file: src/refinement.tex
\subsection{EFG Refinement}
\label{sec:refinement}

An extended event flow graph can be seen as a finite representation of the
regular language $\Lang$ over the alphabet of events $\eventset$ that contains
possible event sequences of the program being analyzed (see
Definition~\ref{def:eventseq}). We transform the extended event flow graph $\EFG
= \langle \locations, \initloc, \efgedge, \labeling \rangle$ into a
non-deterministic finite automaton (NFA) $\AF = \langle\Sigma, Q, q_0, F,
\Delta\rangle$ that accepts the regular language $\Lang$ where:
\begin{itemize}
  \item The input alphabet is the set of events. I.e., $\Sigma = \eventset$
  \item For every location $\location$ in $\EFG$ there is a state
  in the NFA. I.e., \\$Q = \{q_{\location} \suchthat \location \in
  \locations\}$
  \item An initial state $q_0$ is introduced.
  \item All states except for the initial state are accepting states, $F = Q
  \setminus q_0$
  \item The transition relation $\Delta \subseteq Q \times \Sigma \times Q$
  is built from the transition relation of $\EFG$. Every edge in $\efgedge$ has
  a corresponding entry in the transition relation of $\AF$. The successor event
  of the event flow relation is used as input symbol for the edge. Additionally,
  to account for multiple initial locations, an edge is introduced between the
  initial state and the states corresponding to those locations. I.e., 
  \begin{align*}
	    \Delta = & \left\{(q_{\location_1},e_2,q_{\location_2})
	  \suchthat (\location_1, \location_2) \in \efgedge \land \labeling(\location_2)
	  = e_2\right\} \\
	  & {} \cup \left\{ (q_0, e_i, q_{\location_i}) \suchthat \location_i
	  \in \initloc \land \labeling(\location_i) = e_i \right\}
  \end{align*}
  
\end{itemize}

\begin{lemma} \label{lemma:NFAEFG}
The language $\Lang$ accepted by $\AF$ is exactly $\possible(\EFG)$ .
\end{lemma}

\begin{proof}
(sketch) Prove that any event sequence $\seq = e_1, e_2, \ldots, e_n \in
\possible(\EFG)$ is accepted by $\AF$. Apply induction over the prefixes of
$\seq$ and the description of the construction of $\AF$ given $\EFG$.
\end{proof}

A sequence of events $\seq_\infeasible = e_1, e_2, \ldots, e_n$ returned as
non-executable prefix of a violating sequence, is a word accepted by $\AF$ (see
Section~\ref{sec:translation}). If $\seq_\infeasible$ is a non-executable
sequence, we exclude from $\Lang$ the regular language $\infeasible$ of all
sequences prefixed with $\seq_\infeasible$, i.e., sequences of the form
$\seq.\Sigma^*$.

For that purpose, we construct an automaton $\As$ that accepts exactly those
sequences. By using the complement and intersection operations on finite
automata, we obtain a new automaton $\AF^\# = \AF \cap \bar{\infeasible}$ that
accepts exactly the language $\Lang^\# = \Lang \setminus \infeasible$. The new
automaton $\AF^\# = \langle\Sigma, Q^\#, q_0^\#, F^\#, \Delta^\#\rangle$ can be
transformed back into an EEFG $\EFG^\# = \langle \locations^\#, \initloc^\#,
\efgedge^\#, \labeling^\# \rangle^\#$ as follows:

\begin{itemize}
  \item There is a location for each state except for the initial state.\\ I.e.,
  $L^\# = \{\location_q \suchthat q \in Q^\# \setminus q_0^\# \}$
  \item Every location is labeled with the input event of its incoming
  edges.\footnote{By construction, it is guaranteed that all incoming edges to
  a location are labeled with the same event.} I.e.,
  $\labeling^\#(\location) = e$ where $e$ is the input event of all
  incoming edges of $\location$, $(q_i, e, q_\location) \in \Delta^\#$.
  \item The initial locations are those connected to the initial state. \\I.e.,
  $\initloc^\# = \{\location_q \suchthat (q_0^\#, e, q_\location) \in \Delta^\#
  \}$
  \item The transition relation is reconstructed from the transition relation of
  $\AF^\#$:
  \[
  	\efgedge^\# = \{(\location_1,\location_2) \suchthat (q_{\location_1},e,
  	q_{\location_2}) \in \Delta^\# \land q_{\location_1} \neq q_0^\#\}
  \]
\end{itemize}

\begin{lemma} \label{lemma:EFGNFA}
$\possible(\EFG^\#)$ is exactly the language $\Lang^\#$ accepted by $\AF^\#$.
\end{lemma}

\begin{proof}
(sketch) Analogous to Lemma~\ref{lemma:NFAEFG}. Prove that any sequence $\seq \in
\possible(\EFG)$ is accepted by $\AF^\#$ by induction on the prefixes of $\seq$
and the construction of $\EFG^\#$.
\end{proof}
Having defined a translation of event flow graphs to and from NFA enables us to
use the method to provide the refinement step of our algorithm (see
line 9, Algorithm~\ref{alg:verify}).

\begin{theorem}
The possible event sequences of the refined EEFG $\possible(\EFG^\#)$ are the
event sequences of the original EEFG without all sequences that start with
$\seq_\infeasible$. I.e., $\possible(\EFG^\#) = \possible(\EFG^\#) \setminus
\infeasible$.
\end{theorem}

\begin{proof}
(sketch) Use the definition of $\AF^\#$ together with lemmata \ref{lemma:NFAEFG}
and \ref{lemma:EFGNFA} to show: 
 \[\forall \seq \in \infeasible : \seq \in \possible(\EFG) \land \seq \not\in
 \possible(\EFG^\#)\]  
\end{proof}
The refined event flow graph can thus be used for the subsequent iteration of
our verification algorithm.

%% file: src/experiments.tex
\section{Evaluation}
\label{sec:experiments}

We apply our prototype implementation of the verification algorithm presented to
some samples of the benchmark set \texttt{UNL.Toy.2010} from the COMET benchmark
repository\footnote{Available at \url{comet.unl.edu}}. We modified each
benchmark by adding a variable $x$ and assertions to the source code to specify
the reachability of some value, analogous to our motivating example.
For each sample we added an assertion that is never violated, and one where an executable
event sequence leads to a violation.

The results of our experiments are summarized in Table~\ref{table:experiments}.
We report the number of events; the injected assertion; the time needed for the
static analysis; the number of calls to the theorem prover; the number of
refinement iterations (that is, the number of sequences that were replayed by
our approach); the number of sequences of a pure black-box approach; and the final result of the
benchmark.

In the case of an \texttt{UNSAFE} result of the static analysis, we compare the
number of event sequences tested by our approach (the number of iterations) with
the number of event sequences in a pure black-box approach that would be needed
to guarantee correctness of the applications, that is, the detection of a
violated assertion. Note, that in the case of a \texttt{SAFE} result of the
static analysis, a pure black-box approach cannot guarantee the correctness of
the application. Thus, we express the number of event sequences with $\infty$.

\input{src/resulttable}

%% file: src/resulttable.tex
\begin{table}[htbp]
 \centering
\begin{tabular}{l|r|r|r|r|r|r|r}
\hline
Benchmark			&Events	&Assertion	&Time(secs)	&TP Calls &Iterations  &Sequences &Result\\
\hline \hline
repair-2-cons		&3		&x!=8		&1,802		&267		&2			    & $\infty$          &SAFE\\
					&		&x!=9		&0,216		&33			&1			    & 9                 &UNSAFE \\
\hline
repair-2-excl		&3		&x!=3		&14,961		&1482		&4			    & $\infty$          &SAFE \\
					&		&x!=4		&0,066		&33			&1			    & 9                 &UNSAFE \\
\hline
repair-3-cons		&4		&x!=13		&0,028		&25			&1			    & $\infty$          &SAFE \\
					&		&x!=35		&1,019		&119		&1			    & 16                &UNSAFE \\
\hline
repair-3-excl		&5		&x!=13		&0,021		&36			&1			    & $\infty$          &SAFE \\
					&		&x!=25		&0,505		&116		&1			    & 25                &UNSAFE \\
\hline
repair-cmpd			&5		&x!=17		&0,009		&36			&1			    & $\infty$          &SAFE \\
					&		&x!=55		&1,56		&352		&1			    & 25                &UNSAFE\\
\hline
\end{tabular}
\vspace{0.5cm}
\caption{Experimental results.}
\label{table:experiments}
\end{table}

%%% Local Variables:
%%% TeX-master: "../vmcai2013"
%%% End:

%% file: src/related.tex
\section{Related Work}
\label{sec:related}

An approach which comes closest to our work is~\cite{memon99,memon01} which
incorporates planning from the domain of artificial intelligence to generate
test cases for GUI applications. The input to the planning system is a set of
operators (namely, the event handlers), an initial state, and a goal state of
the application. The planning system outputs a sequence of operators that lead
from the initial state to the goal state. However, in this approach a test
engineer has to manually define the preconditions and effects of each operator.
Our approach extends this idea as follows: First, we propose an automatic
translation of the operators of a GUI application into a Boogie program. Second,
the static analysis of our approach can be replaced by other techniques. Third,
our approach replays an event sequence in a black-box fashion on the GUI
application. Fourth, an event sequence which is not executable on the GUI is
used to refine the model which is then again used to translate the GUI
application.

The work in~\cite{berstel01} presents a general approach to specify user
interactions in GUI applications from a design perspective. This technique
allows the analysis of user interactions using model checking, and the synthesis
of user interactions to executable GUI applications. Since the work
in~\cite{berstel01} presents a high-level approach, it obviates the efforts of
extracting models, e.g., from the source code of an existing application. In our
case, we focus on supporting a test engineer which usually deals with executable
GUI applications instead of abstract models. Hence, the translation of an
existing application into a verifiable program presents one of the main
technical contribution of this paper. In particular, our approach allows the
analysis of an executable GUI application, e.g., even in the phase of
\emph{release-to-manufacturing} within a software release life cycle.

An approach which identifies useful abstractions of existing GUI applications is
presented in~\cite{dwyer97}. Those abstractions are based on structural features
of GUI applications, e.g., the enabledness of a button (enabled or disabled)
using a boolean value, or the current value of slider control using an integer
value. First, the abstractions are inferred manually from a GUI application.
Then, the abstractions are used to build a model which is checked by
SMV~\cite{clarke92}. In order to overcome the manual identification
abstractions, the work in~\cite{dwyer04} focuses on the automatic analysis of
interaction orderings with model checking. In the work~\cite{dwyer04} the model
is inferred via analyzing the code statically. The static analysis is tailored
to a specific GUI toolkit, namely Java Swing. Our approach uses a dynamic
approach: a model (the EFG) is created during the execution of the GUI
application. Furthermore, our approach introduces a refinement loop which allows
to improve the initial model by replaying event sequences obtained from the step
static analysis. Since GUI code is written in many ways, a static analysis
technique must be tailored to comprehend the behavior of each GUI toolkit. The
use of a black-box model is justified by the reasonable trade-off between
applicability and precision of a black-box model. Furthermore, the EFG is a
black-box model which works independently from a currently used GUI toolkit.

The work in~\cite{arlt12} (with shared co-authors) presents a lightweight static
analysis, which generates all event sequences that are at the same time
executable and justifiably relevant. First, the approach infers a model which
expresses dependencies of events of the GUI application. Second, event sequences
of bounded length are generated from this dependency model. Third, an event flow
graph is incorporated in order to convert event sequences from the dependency
model into executable event sequences. This paper represents a consistent
further development: it uses an advanced static analysis which is able to reason
about properties of a GUI application, instead of generating all event sequences
that might violate a specific property.

%% file: src/conclusion.tex
\section{Conclusion} % and Future Work}
\label{sec:conclusion}
In this paper, we have presented a method that can give a guarantee of
the absence of faults, i.e., the correctness of \emph{all} test cases
of the black-box model of a GUI application.  We have shown that the
black-model allows us to circumvent the static analysis of the GUI
tool kit, namely by constructing a program  which is amenable to
static analysis.  We have implemented our approach and we have
presented preliminary experiments which indicate its practical
potential.  The technical contribution of the paper is to \emph{make
  things work}.  This includes the construction of the program from a
given EFG as input to a static analysis-based verification procedure
and the program refinement procedure.

The main contribution of the paper is perhaps a conceptual one: the
realization that formal, static analysis-based
verification methods can have an interesting potential for GUI
applications, namely as a supplement to GUI testing.  % Andreas26Aug
The crux is to take the perspective of the GUI tester
who, after having constructed the black-box model of the GUI
application, is faced with what can be phrased as a verification
problem.

In the future, we would like to explore efficient abstractions to support
events with user inputs. In our implementation, an enhanced bytecode translation
could increase the efficiency of static analysis by using a memory model
optimized for the verification of GUI applications.

% - Translation of text boxes
% - Enhanced memory model

%   propose to supplement the test with the verification of the
%   black-box model.  We

% allows us to formulate a verification problem. As we will
% show, it also allows